\documentclass[11pt, a4paper]{amsart}
\usepackage{a4wide}
\usepackage{comment,mathrsfs, amssymb,amsmath,url,amsthm}

\title[The algebraic dichotomy conjecture for infinite domain CSPs]{The algebraic dichotomy conjecture for infinite domain Constraint Satisfaction Problems}
\date{\today}

\author
{Libor Barto}
	\address{Department of Algebra, MFF UK, Sokolovska 83, 186 00 Praha 8, Czech Republic}
	\email{libor.barto@gmail.com}
	\urladdr{http://www.karlin.mff.cuni.cz/~barto/}
\author
{Michael Pinsker}
	\address{Department of Algebra, MFF UK, Sokolovska 83, 186 00 Praha 8, Czech Republic}    
	\email{marula@gmx.at}
    \urladdr{http://dmg.tuwien.ac.at/pinsker/}
\thanks{Libor Barto was supported by the the Grant Agency of the Czech Republic, grant GA\v CR
13-01832S. Michael Pinsker has been funded through project  P27600 of the  Austrian Science Fund (FWF)}

    \newtheorem{theorem}{Theorem}[section]
    \newtheorem{lemma}[theorem]{Lemma}
    
    \newtheorem{conjecture}[theorem]{Conjecture}
    \newtheorem{corollary}[theorem]{Corollary}

    \newtheorem{definition}[theorem]{Definition}

\DeclareMathOperator{\CSP}{CSP}
\DeclareMathOperator{\Pol}{Pol}

\newcommand{\tuple}[1]{\mathbf{#1}}
\newcommand{\clone}[1]{\mathcal{#1}}
\newcommand{\relstr}[1]{\mathbb{#1}}
\newcommand{\NP}{\mathrm{NP}}
\newcommand{\trivclone}{\clone{P}}
\newcommand{\group}[1]{\mathcal{#1}}
\newcommand{\clgr}[1]{\mathrm{Gr}(#1)}
\newcommand{\forb}[1]{\mathscr{#1}}
\newcommand{\To}{\rightarrow}

\begin{document}

\maketitle

\begin{abstract}
We prove that an $\omega$-categorical core structure primitively positively interprets all finite structures with parameters if and only if some stabilizer of its polymorphism clone has a homomorphism to the clone of projections, and that this happens if and only if its polymorphism clone does not contain operations $\alpha$, $\beta$, $s$ satisfying the identity $\alpha s(x,y,x,z,y,z) \approx \beta s(y,x,z,x,z,y)$.

This establishes an algebraic criterion equivalent to the conjectured borderline between P and NP-complete CSPs over reducts of finitely bounded homogenous structures, and accomplishes one of the steps of a proposed strategy for reducing the infinite domain CSP dichotomy conjecture to the finite case.

Our theorem is also of independent mathematical interest, characterizing a topological property of any $\omega$-categorical core structure (the existence of a continuous homomorphism of a stabilizer of its polymorphism clone to the projections) in purely algebraic terms (the failure of an identity as above).
\end{abstract}


\section{Introduction and Main Results}

The Constraint Satisfaction Problem (CSP) over a relational structure $\relstr{A}$ in a finite language, denoted by $\CSP(\relstr{A})$, is the problem of deciding whether or not a given primitive positive (pp-) sentence in the language of $\relstr{A}$ holds in $\relstr{A}$.
An alternative, combinatorial definition of the CSP is also popular in the literature: $\CSP(\relstr{A})$ is the problem of deciding whether a given relational structure in the same language as $\relstr{A}$ maps homomorphically into $\relstr{A}$.

For CSPs over certain structures, including all finite ones, a computational complexity classification has been conjectured, separating NP-hard problems from polynomial-time solvable ones. In the following, we shall state and discuss this conjecture, and subsequently present an improvement thereof which follows from our results. In order to keep the presentation compact, we postpone most definitions to Section~\ref{sect:notation}, and refer also to the monograph~\cite{Bodirsky-HDR} as well as to the shorter~\cite{Pin15}. As a reference for standard notions from model theory, we point to the textbook~\cite{Hodges}.

All structures in the present article are implicitly assumed to be finite or countable. 

\subsection{The tractability conjecture}
The CSP over a structure with finite domain is clearly contained in the class $\NP$. Some well-known $\NP$-complete problems, such as variants of 3-SAT or 3-COLORING, can be formulated as CSPs over suitable finite structures, as well as some problems solvable in polynomial time, such us 2-SAT, HORN-SAT, or systems of linear equations over finite fields. In fact, it has been conjectured~\cite{FederVardi} that CSPs over finite structures enjoy a dichotomy in the sense that every such CSP is either $\NP$-complete, or tractable, i.e., solvable in polynomial time. A large amount of attention has been brought to confirming or refuting this conjecture, resulting in considerable progress; see~\cite{Barto-survey} for a recent brief introduction and survey. In particular, a precise borderline between $\NP$-complete and tractable CSPs has been delineated~\cite{JBK} and is now referred to as the \emph{tractability conjecture} or also the \emph{algebraic dichotomy conjecture}, since most of the equivalent formulations are algebraic.

When we allow the domain of $\relstr{A}$ to be infinite, the situation changes drastically: every computational decision problem is polynomial-time equivalent to $\CSP(\relstr{A})$ for some $\relstr{A}$~\cite{BodirskyGrohe}! A reasonable assumption on $\relstr{A}$ which sends the CSP back to the class $\NP$ and still allows to cover many interesting computational problems which cannot be modeled as the CSP of a finite structure, 
is that $\relstr{A}$ is a reduct of a finitely bounded homogeneous structure. Substantial results for such CSPs include the full complexity classification of the CSPs over the reducts of $(\mathbb{Q};<)$ in \cite{tcsps-journal} (classifying the complexity of problems previously called temporal constraint satisfaction problems), the reducts of the random graph (generalizing Schaefer's theorem for Boolean CSPs to what can be called the propositional logic for graphs), and the reducts of the binary branching $C$-relation~\cite{Phylo-Complexity} (classifying the complexity of problems known as phylogeny CSPs). The methods here include the algebraic methods from the finite, but in addition tools from model theory and Ramsey theory~\cite{BP-reductsRamsey}. Moreover, topological considerations have played a significant role in the development of the theory~\cite{Topo-Birk}, and indeed seem inevitable in a sense, although paradoxically it was believed or at least hoped that they would  ultimately turn out inutile in a general complexity classification. On the other hand, due to the fact that the investigation of infinite domain CSPs is more recent, and the additional technical complications which are to be expected when passing from the finite to the infinite, 
the purely algebraic theory as known in the finite is still quite undeveloped in the infinite; the present work can be seen as the first purely algebraic result for such CSPs.

A generalization of the finite domain tractability conjecture has been formulated by Manuel Bodirsky and the second author. 
To state it, we first recall several basic facts.
When $\relstr{B}$ has a primitive positive (pp-)interpretation without parameters in $\relstr{A}$, then $\CSP(\relstr{B})$ reduces to $\CSP(\relstr{A})$. When $\relstr{A}$ is an $\omega$-categorical core, then this statement is even true for pp-interpretations with parameters. By~\cite{Cores-journal}, every $\omega$-categorical structure, in particular every reduct of a finitely bounded homogeneous structure, is homomorphically equivalent to an $\omega$-categorical core, which is unique up to isomorphism. Moreover,  the CSPs over any two structures which are homomorphically equivalent are equal, and so passing from an $\omega$-categorical structure to its core does not result in any loss of information concerning the CSP. 

These facts imply that the CSP over an $\omega$-categorical structure is $\NP$-hard whenever its core pp-interprets with parameters some structure whose CSP is NP-hard, such as 
\[
\mathbb{K}_3 = (\{1,2,3\}; \neq)\]
 whose CSP is the 3-coloring problem, or
\[
\mathbb{L} = (\{0,1\}; R_{000}, R_{001}, R_{011}, R_{111}), \  R_{abc} = \{0,1\}^3 \setminus \{(a,b,c)\}
\]
whose CSP is the 3-SAT problem, or
\[
\mathbb{M} = (\{0,1\}; \{(0,0,1),(0,1,0),(1,0,0)\})
\]
whose CSP is the positive 1-in-3-SAT problem. In fact, these three structures not only pp-interpret each other, they pp-interpret all finite structures. The infinite domain tractability conjecture postulates, as does the corresponding conjecture for finite structures, that pp-interpreting all finite structures with parameters in the core is the only source of hardness for CSPs over reducts of finitely bounded homogeneous structures.

\begin{conjecture}[Bodirsky + Pinsker\; 2011; cf.~\cite{BPP-projective-homomorphisms}]\label{conj:dicho}
Let $\relstr{B}$ be a reduct of a finitely bounded  homogeneous structure and let $\relstr{A}$ be the core of $\relstr{B}$. Then
\begin{itemize}
\item $\relstr{A}$ pp-interprets all finite structures with parameters (and thus $\CSP(\relstr{B})$ is NP-complete), or
\item $\CSP(\relstr{B})$ is solvable in polynomial time.
\end{itemize}
\end{conjecture}

In the present article, we show that the failure of the first condition of this conjecture is witnessed by a certain algebraic fact that could, similarly to what is the hope in the finite setting, potentially be exploited for proving tractability of the CSP. We now make the notion of an algebraic witness more precise.

\subsection{The algebraic approach}

The algebraic approach to finite domain CSPs is based on the fact that pp-interpretability strength of a finite structure $\relstr{A}$ is determined by its set of compatible operations, the so called \emph{polymorphism clone} of $\relstr{A}$, denoted by $\Pol(\relstr{A})$. Namely, by classical universal algebraic results, 
a finite structure $\relstr{A}$ pp-interprets a finite structure $\relstr{B}$ if and only if there exists a clone homomorphism from $\Pol(\relstr{A})$ to $\Pol(\relstr{B})$, that is, a mapping which preserves arities and \emph{identities} (universally quantified equations). This fact implies that the complexity of $\CSP(\relstr{A})$ only depends on the identities satisfied by operations in $\Pol(\relstr{A})$  
and leads to an algebraic reformulation of the first item of Conjecture~\ref{conj:dicho} for finite structures. 
The following notation is used: the \emph{stabilizer} of $\Pol(\relstr{A})$ by finitely many constants $c_1,\ldots,c_n$ is denoted $\Pol(\relstr{A},c_1,\ldots,c_n)$; its elements are those polymorphisms of $\relstr{A}$ which preserve all unary relations $\{c_i\}$. The clones $\Pol(\relstr{L})$, $\Pol(\relstr{M})$, as well as $\Pol(\relstr{K},0,1,2)$ are trivial, i.e., they contain only projections. Let us denote the clone of projections on a $2$-element set by $\trivclone$. The clone of projections on any other set of at least $2$ elements is isomorphic to $\trivclone$.

\begin{theorem}[\cite{Geiger,BoKaKoRo,Bir-On-the-structure}, cf.~\cite{Bodirsky-HDR}] \label{thm:finite-birk}
The following are equivalent for a finite relational structure $\relstr{A}$ with domain $A = \{c_1, \dots, c_n\}$.
\begin{itemize}
\item $\relstr{A}$ pp-interprets all finite structures with parameters.
\item There exists a clone homomorphism from $\Pol(\relstr{A},c_1,\ldots,c_n)$ to $\trivclone$. 
\end{itemize}
\end{theorem}

For the second, algebraic statement of Theorem~\ref{thm:finite-birk} numerous equivalent algebraic criteria have been obtained within the setting of finite structures~\cite{T77,HM88,JBK,MM08,Sig10,Cyclic}, making in particular the failure of the condition more easily verifiable: this failure is then usually witnessed by the satisfaction of particular identities in $\Pol(\relstr{A},c_1,\ldots,c_n)$ which cannot be satisfied in $\trivclone$.

Some of the above-mentioned facts about finite domain CSPs have analogues for $\omega$-categorical structures. The complexity of $\CSP(\relstr{A})$ still only depends on the polymorphism clone $\Pol(\relstr{A})$ \cite{BodirskyNesetrilJLC}, and there is an analogue of Theorem~\ref{thm:finite-birk}, which however takes into consideration the natural topological structure of $\Pol(\relstr{A})$.

\begin{theorem}[\cite{Topo-Birk}]\label{thm:int}
The following are equivalent for an $\omega$-categorical structure $\relstr{A}$.
\begin{itemize}
\item $\relstr{A}$ pp-interprets all finite structures with parameters.
\item There exists a \emph{continuous} clone homomorphism from $\Pol(\relstr{A},$ $c_1,\ldots,c_n)$ to $\trivclone$, for some $c_1,\ldots,c_n\in A$.
\end{itemize}
\end{theorem}

More generally, the complexity of $\CSP(\relstr{A})$ for $\omega$-categorical structures provably only depends on the structure of $\Pol(\relstr{A})$ as a topological clone~\cite{Topo-Birk}. A natural, yet unresolved problem when comparing the finite with the $\omega$-categorical setting then
is whether the topological structure of the polymorphism clone is really essential in the infinite, or whether the abstract algebraic structure, i.e., the identities that hold in $\Pol(\relstr{A})$, is sufficient to determine the complexity of the CSP. 

\subsection{The result}
We show that the borderline proposed in Conjecture~\ref{conj:dicho} \emph{is} purely algebraic. In particular, if the conjecture is true, then the complexity of CSPs over structures concerned by the conjecture only depends on the identities which hold in the polymorphism clone of their core, rather than the additional topological structure thereof. Moreover, the borderline is characterized by a single
 simple identity generalizing that of~\cite{Sig10}. We show the following.

\begin{theorem}\label{thm:main}
The following are equivalent for an $\omega$-categorical core structure $\relstr{A}$. 
\begin{itemize}
\item[(i)] There exists no continuous clone homomorphism $\Pol(\relstr{A},c_1,\ldots,$ $c_n) \To\trivclone$, for any $c_1,\ldots,c_n\in\relstr{A}$.
\item[(ii)] There exists no clone homomorphism $\Pol(\relstr{A},c_1,\ldots,c_n)\To\trivclone$, for any $c_1,\ldots,c_n\in\relstr{A}$.
\item[(iii)] $\Pol(\relstr{A})$ contains a pseudo-Siggers operation, i.e., a 6-ary operation $s$ such that 
\[
\alpha s(x,y,x,z,y,z) \approx \beta s(y,x,z,x,z,y)
\] 
for some unary operations $\alpha, \beta \in \Pol(\relstr{A})$.
\end{itemize}
\end{theorem}

Consequently, the missing piece for proving Conjecture~\ref{conj:dicho} can now be stated in purely algebraic terms. 

\begin{conjecture} \label{conj:newdicho}
Let $\relstr{A}$ be the core of a reduct of a finitely bounded  homogeneous structure.
If $\Pol(\relstr{A})$ contains a pseudo-Siggers operation, then $\CSP(\relstr{A})$ is solvable in polynomial time. 
\end{conjecture}

In a proposed strategy~\cite{Pin15} for solving Conjecture~\ref{conj:dicho}, the first step asked to prove that for an $\omega$-categorical structure $\relstr{A}$, the existence of a clone homomorphism $\Pol(\relstr{A})\To\trivclone$ implies the existence of a continuous such homomorphism (cf.~\cite{BPP-projective-homomorphisms}). If this was true, then the failure of the first item of Conjecture~\ref{conj:dicho} would have an algebraic witness, i.e., a non-trivial identity holding in some polymorphism clone $\Pol(\relstr{A},c_1,\ldots,c_n)$. The idea is to then, roughly speaking, ``lift" the algorithm for finite structures whose polymorphism clone satisfies this identity (assuming the finite tractability conjecture is true) via Ramsey theory to show that $\CSP(\relstr{A})$ is tractable.

While we do not answer this question, Theorem~\ref{thm:main} gives an answer for the variant which is actually relevant for the CSP: for an $\omega$-categorical core structure $\relstr A$, the existence of a clone homomorphism $\Pol(\relstr{A})\To\trivclone$ implies the existence of a continuous clone homomorphism from \emph{some stabilizer} of $\Pol(\relstr{A})$ to $\trivclone$. Taking into account the existence of non-continuous clone homomorphisms $\Pol(\relstr{A})\To\trivclone$~\cite{BPP-projective-homomorphisms}, for an $\omega$-categorical $\relstr{A}$, as well as the recent discovery of $\omega$-categorical structures $\relstr{A}, \relstr{A}'$ whose polymorphism clones are isomorphic algebraically, but not topologically~\cite{BEKP}, 
it might very well turn out that the answer to the original question is negative, but, as we would then see a posteriori, irrelevant for CSPs.

Let us also remark that Theorem~\ref{thm:main} is, by the fact that every $\omega$-categorical structure has a unique $\omega$-categorical core, a statement about all $\omega$-categorical structures, rather than only the structures concerned by Conjecture~\ref{conj:dicho}. Theorem~\ref{thm:main} is therefore remarkable in that non-trivial statements about the class of all $\omega$-categorical structures, other than the fundamental theorem of Ryll-Nardzewski, Engeler, and Svenonius characterizing them,  are practically non-existent.

\subsection{Outline and proof strategy}

The strategy for proving Theorem~\ref{thm:main} is similar to the finite analogue of Theorem~\ref{thm:main} proved in~\cite{Sig10} (see also~\cite{KMM14}).
Siggers's reasoning is based on a ``loop lemma'' from Bulatov's paper~\cite{B05} that refines the dichotomy theorem for finite undirected graphs~\cite{HellNesetril}. 

After providing definitions and notation in Section~\ref{sect:notation}, we start our proof in Section~\ref{sect:loop} with a generalization of the loop lemma, the \emph{pseudoloop lemma}, using some of the ideas from~\cite{B05}. Instead of finite graphs we work with infinite objects which we call graph-group-systems, and which can be imagined as an infinite permutation group acting on an infinite fuzzy graph, which is a finite graph modulo the action. Theorem~\ref{thm:main} is then derived from the pseudoloop lemma in Section~\ref{sect:main} using a standard universal algebra technique adapted to the $\omega$-categorical setting via a compactness argument. 

Section~\ref{sect:disc} at the end of the article contains further discussion of our results in the light of other recent results, in particular from the wonderland of reflections~\cite{wonderland}, as well as inspiration for future work.

\section{Definitions and Notation}\label{sect:notation}

Relational structures are denoted by blackboard bold letters, such as $\relstr{A}$, and their domain by the same letter in the plain font, such as $A$.
By a \emph{graph} we mean a relational structure with a single symmetric binary relation. 

\subsection{The range of the infinite CSP conjecture}  
A relational structure $\relstr{B}$ is \emph{homogeneous} if every isomorphism between finite induced substructures extends to an automorphism of the entire structure $\relstr{B}$. In that case, $\relstr{B}$ is uniquely determined, up to isomorphism, by its \emph{age}, i.e., the class of its finite induced substructures up to isomorphism. $\relstr{B}$ is \emph{finitely bounded} if its signature is finite and its age is given by a finite set $\forb{F}$ of forbidden finite substructures, i.e., the age consists precisely of those finite structures in its signature which do not embed any member of $\forb{F}$. A \emph{reduct} of a structure $\relstr{B}$ is a structure $\relstr{A}$ on the same domain which is first-order definable without parameters in $\relstr{B}$. Reducts $\relstr{A}$ of finitely bounded homogeneous structures are \emph{$\omega$-categorical}, i.e., the up to isomorphism unique countable model of their first-order theory. Equivalently, their automorphism groups are \emph{oligomorphic}: they have finitely many orbits in their action on $n$-tuples over $\relstr{A}$, for every finite $n\geq 1$.

\subsection{ pp-formulas and interpretations} A formula is \emph{primitive positive}, in short \emph{pp}, if it contains only equalities, existential quantifiers, conjunctions, and atomic formulas -- in our case, relational symbols. A pp-formula \emph{with parameters} can contain, in addition, elements of the domain. 

A \emph{pp-interpretation} is a first-order interpretation in the sense of model theory where all the involved formulas are primitive positive: 
a structure $\relstr{A}$ \emph{pp-interprets} $\relstr{B}$ if there exists a partial mapping $f$ from a finite power $A^n$ to $B$ such that
the domain of $f$, the $f$-preimage of the equality relation and the $f$-preimage of every relation in $\relstr{B}$ is pp-definable in $\relstr{A}$. In particular, $\relstr{A}$ pp-interprets its substructures induced by pp-definable subsets and also its quotients modulo a pp-definable equivalence relation. 

\subsection{Cores} An $\omega$-categorical structure $\relstr{A}$ is a \emph{core}, also called \emph{model-complete core}, if all of its endomorphisms are elementary self-embeddings, i.e., preserve all first-order formulas over the structure. This is the case if and only if its automorphism group is dense in its endomorphism monoid with respect to the pointwise convergence topology on functions on $A$; cf.~Section~\ref{sect:top} for a description of the latter. Two structures $\relstr{A}, \relstr{B}$ are \emph{homomorphically equivalent} if there exist homomorphisms from $\relstr{A}$ into $\relstr{B}$ and vice-versa.

\subsection{Clones} 
A \emph{function clone} $\clone{C}$ is a set of finitary operations on a fixed set $C$ which contains all projections and which is closed under composition. 
A \emph{polymorphism} of  a relational structure $\relstr{A}$ is a finitary operation $f(x_1,\ldots,x_n)$ on $A$  which \emph{preserves} all relations $R$ of $\relstr{A}$: this means that for all $\tuple{r}_1,\ldots,\tuple{r}_n\in R$ we have that $f(\tuple{r}_1,\ldots,\tuple{r}_n)$, calculated componentwise, is again in $R$. The \emph{polymorphism clone} of $\relstr{A}$, denoted by $\Pol(\relstr{A})$, consists of all polymorphisms of $\relstr{A}$, and is always a function clone. Its unary operations are precisely the endomorphisms of $\relstr{A}$, and its invertible unary operations are precisely the automorphisms of $\relstr{A}$.

A \emph{clone homomorphism} is a mapping from one function clone to another which preserves arities, composition, and which sends every projection of its domain to the corresponding projection of its co-domain. Clone homomorphisms preserve all identities which hold in a function clone, as defined in the introduction.

\subsection{Topology}\label{sect:top}

Function clones carry a natural topology, the topology of \emph{pointwise convergence}, for which a subbasis is given by sets of functions which agree on a fixed finite tuple; the functions of a fixed arity in a function clone form a clopen set. Equivalently, the domain of a function clone is taken to be discrete, and the $n$-ary functions in the clone equipped with the product topology, for every $n\geq 1$; the whole clone is then the sum space of the spaces of $n$-ary functions.

We always understand continuity of clone homomorphisms with respect to this topology.

The function clones which are closed in the topology of pointwise convergence are precisely the polymorphism clones of relational structures. 

\subsection{Core clones and oligomorphicity}

We say that a closed function clone is a \emph{core} 
if  it is the polymorphism clone of a core. 

A function clone $\clone{C}$ is \emph{oligomorphic} if and only if  the permutation group $\clgr{\clone{C}}$ of unary invertible elements of $\clone{C}$ is oligomorphic. When $\clone{C}$ is closed, then this is the case if and only if it is the polymorphism clone of an $\omega$-categorical structure, in which case $\clgr{\clone{C}}$ consists of the automorphisms of that structure. 

Note that when $\clone{C}$ is a core, then the set of its unary operations is the closure of $\clgr{\clone{C}}$.
 

\subsection{Pseudo-Siggers operations} A 6-ary operation $s$ in a function clone $\clone{C}$ is a \emph{pseudo-Siggers operation} if there exist unary $\alpha, \beta \in \clone{C}$ 
such that $\alpha s(x,y,x,z,y,z) \approx \beta s(y,x,z,x,z,y)$ holds  (where $\approx$ means that equality holds for all values for the variables in $C$). We then also say that $s$ \emph{satisfies} the pseudo-Siggers identity.

\section{The Pseudoloop Lemma}\label{sect:loop}

The following definition is a generalization of finite graphs to the $\omega$-categorical which is suitable for our purposes.

\begin{definition}\label{defn:gg}
A \emph{graph-group-system}, in short \emph{gg-system}, is a pair $(\relstr{G},\group{G})$, where $\group{G}$ is a permutation group on a set $G$, and $\relstr{G}=(G;R)$ an (undirected) graph which is invariant under $\group{G}$. We also write $(R,\group{G})$ for the same gg-system. 

The system is called \emph{oligomorphic} if $\group{G}$ is; in that case, $\relstr{G}$ is $\omega$-categorical, since its automorphism group contains $\group{G}$ and hence is oligomorphic. 

The system \emph{pp-interprets (pp-defines)} a structure $\relstr{B}$ if $\relstr{G}$ together with the orbits of $\group{G}$ on finite tuples does.

A \emph{pseudoloop} of a gg-system $(\relstr{G},\group{G})$ is an edge of $\relstr{G}$ of the form $(a,\alpha(a))$, where $\alpha\in\group{G}$. 
\end{definition}

Note that $R$, as well as any relation that is first-order definable from a gg-system $(\relstr{G},\group{G})$, is invariant under the natural action of $\group{G}$ on tuples. In particular, such relations are unions of orbits of the action of $\group{G}$ on tuples, and 
when $\group{G}$ is oligomorphic, then there are only finite many first-order definable relations of any fixed arity. 

We are now ready to state our pseudoloop lemma for gg-systems.

\begin{lemma}[The pseudoloop lemma]\label{lem:loop}
Let $(\relstr{G},\group{G})$ be an oligomorphic gg-system, where $\relstr{G}$ has a subgraph isomorphic to $\relstr{K}_3$. Then either it {pp}-interprets $\relstr{K}_3$ with parameters, or it contains a pseudoloop.
\end{lemma}

For the proof of Lemma~\ref{lem:loop}, we need the following auxiliary definitions.

\begin{definition}
Let $(\relstr{G},\group{G})$ be a gg-system. The \emph{support} of $(\relstr{G},\group{G})$ are those elements of its domain which are contained in an edge.

For $a_1,\ldots,a_n \in G$, we denote by $O(a_1,\ldots,a_n)$ the orbit of the tuple $(a_1,\ldots,a_n)$ under $\group{G}$.
\end{definition}

The support of a gg-system $(\relstr{G},\group{G})$ is a union of $\group{G}$-orbits by the remark below Definition~\ref{defn:gg}. This justifies the following definition.

\begin{definition}
A pseudoloop-free gg-system $(\relstr{G},\group{G})$ containing a $\relstr{K}_3$ is \emph{minimal} if it does not pp-define a symmetric relation $R'$ on $G$ such that $(R',\group{G})$ is a pseudoloop-free gg-system containing a $\relstr{K}_3$ whose support consists of fewer orbits. 
\end{definition}

We can now prove the pseudoloop lemma.

\begin{proof}[Proof of Lemma~\ref{lem:loop}]
Assuming that a gg-system $(\relstr{G},\group{G})$, where $\relstr{G} = (G;R)$, has no pseudoloop, we show that it {pp}-interprets $\relstr{K}_3$ with parameters. \bigskip

\noindent \textit{Step 0}: If $(\relstr{G},\group{G})$ is not minimal, then we can replace it by a minimal gg-system. We thus henceforth assume that it is minimal.\bigskip

\noindent \textit{Step 1}: $R$ pp-defines a symmetric binary relation $R'$ with the property that every edge of $R'$ is contained in a $\relstr{K}_3$, i.e.,  every element of $R'$ is contained in an induced  subgraph of $(G;R')$ isomorphic to $\relstr{K}_3$, and which still shares our assumptions on $R$: 
$$
R'(x,y):\leftrightarrow \exists z\; R(x,y)\wedge R(x,z)\wedge R(y,z).
$$
Hence, replacing $R$ by $R'$, we henceforth assume that every edge of $R$ is contained in a $\relstr{K}_3$.\bigskip

\noindent 
In the following, for $n\geq 1$ we say that $x,y\in G$ are \emph{$n$-diamond-connected}, denoted by $x\sim_n y$, if there exist $a_1$, $b_1$, $c_1$, $d_1$, $\ldots$, $a_n$, $b_n$, $c_n$, $d_n\in G$ such that, for every $1\leq i\leq n$, both $a_i,b_i,c_i$ and $b_i,c_i,d_i$ induce $\relstr{K}_3$, 
$x=a_1$, $d_1=a_2$, $d_2=a_3$, \dots, $d_{n-1} = a_{n}$, and $d_n=y$. They are \emph{diamond-connected}, denoted by $x\sim y$, if they are $n$-diamond-connected for some $n\geq 1$. 

Observe that $\sim_n$  is a pp-definable relation from $R$ (since our definition is in fact a pp-definition). Also recall that there are only finitely many binary relations first-order definable from $R$, and note that if $x,y$ are $n$-diamond-connected, then they are $m$-diamond-connected for all $m\geq n$. Therefore, there exists an $n\geq 1$ such that $x,y$ are diamond-connected if and only if they are $n$-diamond connected. In particular, the relation $x\sim y$ is pp-definable in $\relstr{G}$. Note also that it is an equivalence relation on the support of $R$: it is clearly transitive and symmetric, and it is reflexive since on the support every vertex is contained in a $\relstr{K}_3$, by Step~1.\bigskip

\noindent \textit{Step 2}: We claim that if $x, y\in G$ are $n$-diamond-connected for some $n \geq 1$, then $\neg R(x,y')$ for all $y'\in O(y)$. Otherwise, pick a counterexample $x,y,y'$ with minimal $n\geq 1$. 

Suppose first that $n$ is odd and set $k := \frac{n-1}{2}$. Let $a$ be the $a_{k+1}$ from the chain of diamonds witnessing $x\sim_n y$. 
Consider the following pp-definition over $(\relstr{G},\group{G})$:
$$
S(w):\leftrightarrow \exists u,v\;(u\in O(a)\wedge u\sim_{k} v\wedge R(v,w))\; ;
$$
in case that $k=0$ we replace $\sim_k$ by the equality relation. Then clearly $S(b_n)$ and $S(c_n)$. But we also have $S(y)$, since $S(y')$ holds by virtue of $a\sim_k x$ and $R(x,y')$ and since $y$ is in the same orbit as $y'$. 
Hence, since $d_n=y$, we have $S(d_n)$ and so $S$ contains a $\relstr{K}_3$. 
By the minimality of $(\relstr{G},\group{G})$ (see Step~0), it contains $x$, for otherwise we could intersect $R$ with $S^2$ and obtain a relation whose support consists of a smaller number of orbits. Let $u, v\in G$ as in the definition of $S$ witness that $S(x)$ holds. Then $u\sim _{k} v$, but also $u\sim_{k} x'$ for some $x'\in O(x)$, as $a \sim_k x$, $u \in O(a)$, and $\sim_k$ is invariant under $\group{G}$. Therefore, $v\sim_{n-1} x'$, which together with  $R(v,x)$ contradicts the minimality of $n$ when $n\geq 3$; when $n=1$, this means that we have discovered a pseudoloop of $(\relstr{G},\group{G})$, again a contradiction.

Suppose now that $n$ is even and denote $k = \frac{n}{2}-1$; the argument is similar. 
Let $b,c$ be the $b_{k+1}, c_{k+1}$ from the chain of diamonds witnessing $x\sim_n y$. 
Consider the following pp-definition:
\begin{align*}
S(w):\leftrightarrow \exists & u_b,u_c,u,v\;((u_b,u_c)\in O(b,c)\wedge \\
  & R(u,u_c)\wedge R(u,u_b)\wedge u\sim_{k} v\wedge R(v,w))\; .
\end{align*}
Then as in the odd case, $S(b_n), S(c_n), S(d_n)$, and so the set defined by $S$ contains a $\relstr{K}_3$. 
By the minimality of $(\relstr{G},\group{G})$, it contains $x$; let $u_b,u_c,u,v\in G$ as in the definition of $S$ witness this. Then $u\sim _{k} v$, but also 
 $u\sim _{k+1} x'$ for some $x'\in O(y)$. Hence, $v\sim_{n-1} x'$ and $R(v,x)$ contradict the minimality of $n$.\bigskip
 
 \noindent \textit{Step 3}: Defining 
 $$
 R'(x,y):\leftrightarrow \exists x', y'\; (x\sim x'\wedge y\sim y'\wedge R(x',y'))
 $$
 we obtain a relation $R'\supseteq R$ which does not contain a pseudoloop. Indeed,
if $R'(x,y)$ is witnessed by $x',y'$ and $x$ and $y$ are in the same orbit, then 
$x \sim y''$ for some $y'' \in O(y')$ since $y\in O(x)$ and since $\sim$ is invariant under $\group{G}$. Thus $x' \sim y''$ and $R(x',y')$, a contradiction with Step~2.
Moreover, every edge in $R'$ is contained in a $\relstr{K}_3$: if $z'$ is so that $\{x',y',z'\}$ induce a $\relstr{K}_3$ in $R$, 
then $\{x,y,z'\}$ induce a $\relstr{K}_3$ in $R'$, for $z'\sim z'$ and $x\sim x'$ imply $R'(x,z')$, and similarly we infer $R'(y,z')$. 

Therefore, we may replace $R$ by $R'$. If this replacement changes the equivalence $\sim$, we repeat Step 3. 
Since the relation $R$ gets bigger after each step and $\group{G}$ is oligomorphic, this process stabilizes after finitely many steps.
\bigskip

\noindent \textit{Step 4}:
Now $R$ is in fact a relation between equivalence classes of $\sim$ and the naturally defined quotient gg-system $(\relstr{G}^q,\group{G}^q)$ on $G^q = G/\sim$ contains neither pseudoloops (by Step~3), nor \emph{diamonds}, that is, there do not exist distinct $a,b,c,d \in G^q$ such that $\{a,b,c\}$ and $\{b,c,d\}$ both induce a $\relstr{K}_3$. 
Moreover, every edge of $\relstr{G}^q$ is still contained in a $\relstr{K}_3$. 
We replace the gg-system $(\relstr{G},\group{G})$ by $(\relstr{G}^q,\group{G}^q)$.

If this new gg-system was not minimal, we could repeat the whole proof starting from Step 0. In each reiteration, Step 0 decreases the number of orbits in the support, while no other step increases it, so this process will terminate after a finite number of steps, and we will end up with a minimal gg-system after Step~4 eventually. The first author insists to remark that this process is in fact unnecessary, as the gg-system is already minimal after the first round of Steps~0 to~4: if after factoring it did pp-define a relation whose support consists of less $\group{G}^q$-orbits,
then the syntactically same pp-definition would show that the original gg-system was not minimal.

\bigskip

\noindent
Whichever solution we prefer, summarizing we end up with a gg-system $(\relstr{G},\group{G})$ which is minimal, pseudoloop-free, and diamond-free; still, every edge of $\relstr{G}$ is contained in a $\relstr{K}_3$. 
\bigskip

\noindent \textit{Step 5}: For $k\geq 1$, we denote the $k$-th power of $\relstr{K}_3$ by $\relstr{T}_k$.
By Lemma~\ref{lem:Tk_bound} shown below, there exists an $m\geq 1$ such that, for any $k \geq m$, $\relstr{G}$ has no induced subgraph isomorphic to $\relstr{T}_k$.\bigskip

\noindent \textit{Step 6}: Recall that $\relstr{G}$ contains $\relstr{T}_1 = \relstr{K}_3$. By Step 5, there exists a maximal $k\geq 1$ such that $\relstr{G}$ contains an induced subgraph isomorphic to $\relstr{T}_k$.
Let $k$ be that number and let $a_1, \dots, a_l$, where $l = |T_k| = 3^{k}$, denote the vertices of such an induced subgraph. 
We show that $\relstr{G}$ pp-defines the set $A = \{a_1, \dots, a_l\}$ with parameters $a_1,\ldots,a_l$. 

By~\cite{BodirskyNesetrilJLC}, this is the case if each $l$-ary operation $f$ in $\Pol(\relstr{G},a_1, \dots, a_l)$ preserves $A$. So, suppose that such a function $f$ does not preserve $A$.
Now, $f$ is a homomorphism $\relstr{G}^l \To \relstr{G}$ and its restriction to $A$ is a homomorphism $f'$ from $\relstr{T}_k$ to the diamond-free graph $\relstr{G}$ whose image, which contains $A$ because $f$ stabilizes each $a_i$, is strictly larger than $|T_k|$. \cite[Claim 3, Subsection 3.2]{B05} shows that the image of $f'$ induces a graph isomorphic to $\relstr{T}_m$ for some $m > k$, a contradiction.\bigskip

\noindent \textit{Step 7}: Step 6 implies that $\relstr{G}$ pp-interprets $\relstr{T}_k$ with parameters. But $\relstr{K}_3$ can be pp-interpreted in $\relstr{T}_k$ with parameters by the final sentence of~\cite{B05}.
\end{proof}

\begin{lemma}\label{lem:Tk_bound}
Let $(\relstr{G},\group{G})$ be an oligomorphic gg-system containing a $\relstr{K}_3$ and having no pseudoloops, and assume the system is minimal. 
Then $\relstr{G}$ has no induced subgraph isomorphic to $\relstr{T}_k$ for any $k\geq 1$ where $|T_k|=3^{k}$  exceeds the number of orbits of $\group{G}$.\bigskip
\end{lemma}
\begin{proof}
Suppose there exists a counterexample $(\relstr{G},\group{G})$, where $\relstr{G}=(G;R)$. 
We may assume without loss of generality that every edge of $\relstr{G}$ is contained in a $\relstr{K}_3$ by keeping only those which are, as in Step~1 of the proof of Lemma~\ref{lem:loop}.
Under those assumptions, we pp-define in $(\relstr{G},\group{G})$ a symmetric relation $R'\supsetneq R$ without pseudoloops in which every edge is still contained in a $\relstr{K}_3$, and such that the support of $(R',\group{G})$ equals the support of $(R,\group{G})$. Repeating this process, by oligomorphicity we must after finitely many steps arrive at a gg-system $(R',\group{G})$ which is not minimal. Hence, $(\relstr{G},\group{G})$ was not minimal in the first place, a contradiction. 

Fix a copy of $\relstr{T}_k$ in $\relstr{G}$, the elements of which we denote by tuples in $\{1,2,3\}^k$. So, two vertices in $\{1,2,3\}^k$ are adjacent if and only if they differ in every coordinate. 
From the cardinality assumption, we can pick two elements $\tuple{a},\tuple{a}'$ of the copy that belong to the same orbit $A$. Let $\tuple{b},\tuple{c}$ in the copy be so that $\{\tuple{a},\tuple{b},\tuple{c}\}$ induce a $\relstr{K}_3$, and let $B,C$ be their orbits. Since $(\relstr{G},\group{G})$ has no pseudoloops, the three orbits $A,B,C$ are distinct. Without loss of generality, assume $\tuple{a}=1^k$ (i.e., the $k$ tuple all of whose entries equal $1$), $\tuple{b}=2^k$, and $\tuple{c}=3^k$.

 Define a relation
\begin{align*}
	S(u,v):\leftrightarrow \exists & a'',b'',c'', n_A,n_B,n_C \; \\
	&(R(u,n_A)\wedge R(v,n_A)\wedge R(n_A,a'')\wedge a''\in A\; \wedge\\
	& R(u,n_B)\wedge R(v,n_B)\wedge R(n_B,b'')\wedge b''\in B\; \wedge\\
	& R(u,n_C)\wedge R(v,n_C)\wedge R(n_C,c'')\wedge c''\in C)\; .
\end{align*}
In words, $u,v$ have common neighbors adjacent to elements in $A, B,$ and $C$. 

The relation $S$ is obviously symmetric. It is also reflexive on the support of $(\relstr{G},\group{G})$: every element of the support is a neighbor of a neighbor of an element in $A$, and similarly in $B$ and $C$, by the minimality of $(\relstr{G},\group{G})$; otherwise, we could restrict $R$ to neighbors of neighbors of $A$, a set which contains $A\cup B\cup C$; we would therefore obtain a smaller support gg-system containing a $\relstr{K}_3$, namely the one induced by $\{\tuple{a},\tuple{b},\tuple{c}\}$. 

Observe that whenever $S(u,v)$ holds, then every element of the support of  $(\relstr{G},\group{G})$ is adjacent to a common neighbor of $O(u)$ and $O(v)$: this follows as above from the minimality of $(\relstr{G},\group{G})$ since the elements of  $A\cup B \cup C$ are adjacent to a common neighbor of $O(u)$ and $O(v)$.

Set
$$
Q(u,v):\leftrightarrow \exists s\; (R(u,s)\wedge S(s,v))\; \wedge\; \exists t\; (S(u,t)\wedge R(t,v))\; .
$$
Then $Q\supseteq R$: since $S$ is reflexive on the support of $(\relstr{G},\group{G})$, setting $s=v$ and $t=u$ in the above definition shows that $R(u,v)$ implies $Q(u,v)$. Moreover, $Q$ is symmetric by definition. Let $R'$ consist of those edges of $Q$ which are contained in a $\relstr{K}_3$ with respect to $Q$. We still have that $R'\supseteq R$.

We now show that $(Q,\group{G})$, and thus $(R',\group{G})$, has no pseudoloop. To this end, it suffices to show that whenever $R(u,v)$ holds, then we cannot have $S(u,v')$ for any $v'\in O(v)$. Suppose to the contrary that there exist such elements. 
The $R$-edge $(u,v)$ is contained in a $\relstr{K}_3$, induced by $\{u,v,w\}$, for some $w\in G$.
As observed above, each vertex, in particular the vertex $w$, is adjacent to a common neighbor of $O(u)$ and $O(v')=O(v)$. 
Therefore, there exists a common neighbor $z$ of $O(u)$, $O(v)$ and $O(w)$. 
The set of neighbors of $O(z)$ contains $O(u), O(v)$, and $O(w)$; it is a proper subset of $G$ since $(\relstr{G},\group{G})$ has no pseudoloops; it is pp-definable in $(\relstr{G},\group{G})$; and finally, it contains a $\relstr{K}_3$, contradicting the minimality of $(\relstr{G},\group{G})$.

Using for the first time the copy of $\relstr{T}_k$ in $\relstr{G}$, we now show that $R$ is properly contained in $R'$ by showing that $\tuple{a}'$, the second element of the copy of $\relstr{T}_k$ in the orbit $A$ of $\tuple{a}$, is related to $\tuple{b}$ and $\tuple{c}$ via $R'$. Note that this is sufficient since in $\relstr{T}_k$, no two distinct elements are related to both $\tuple{b}$ and $\tuple{c}$. 
We show only $R'(\tuple{a}',\tuple{b})$, the second claim is analogous. 
Reordering the tuples when necessary, we may assume that $a'_i\neq 2$ for all $1\leq i\leq j$, and $a'_i= 2$ for all $j<i\leq k$. Since $\tuple{a}'\neq \tuple{b}$, we have $j\geq 1$. 
Observe that whenever $\tuple{u}, \tuple{v} \in \{1,2,3\}^k$ are of the form $(x,\ldots,x,2,\ldots,2)$ and $(x,\ldots,x,3,\ldots,3)$, respectively, where the number of occurences of $x$ equals $j$, then $S(\tuple{u},\tuple{v})$: this is witnessed by their common neighbor $(y_1,\ldots,y_j,1,\ldots,1)$, where $y_i\notin\{a'_i,x\}$ for all $1\leq i\leq j$, which is $R$-related to $\tuple{a}'\in A$; their common neighbor $(z,\ldots,z,1,\ldots,1)$, starting with $j$ occurrences of $z\notin\{2,x\}$, which is $R$-related to $\tuple{b}\in B$; and their common neighbor $(w,\ldots,w,1,\ldots,1)$, starting with $j$ occurrences of $w\notin\{3,x\}$, which is $R$-related to $\tuple{c}\in C$. But now we see that $Q(a',b)$ holds: setting $\tuple{t} = (a_1',\ldots,a_j',3,\ldots,3)$, 
we have $S(\tuple{a}',\tuple{t})$ and $R(\tuple{t},\tuple{b})$; on the other hand, setting $\tuple{s}:=(2,\ldots,2,3,\ldots,3)$, with $j$ occurrences of $2$, we have $R(\tuple{a}',\tuple{s})$ and $S(\tuple{s},\tuple{b})$. We can then conclude that $R'(\tuple{a}',\tuple{b})$ holds, since any two elements of $\{1,2,3\}^k$, in particular $\tuple{a}'$ and $\tuple{b}$, have a common neighbor with respect to $R$, and hence also with respect to $Q$, showing that the $Q$-edge $(\tuple{a}',\tuple{b})$ is contained in a $\relstr{K}_3$ with respect to $Q$.

\end{proof}

\section{Proof of the main result} \label{sect:main}

In order to derive Theorem~\ref{thm:main}, we will produce pseudo-Siggers operations locally using the pseudoloop lemma, and then derive a global pseudo-Siggers operation via a compactness argument.

\begin{definition}
We say that a function clone $\clone{C}$ has \emph{local pseudo-Siggers operations} if for every finite $A\subseteq C$ there exists a $6$-ary $s\in\clone{C}$ and unary $\alpha, \beta \in\clone{C}$ satisfying 
\[
\alpha s(x,y,x,z,y,z) = \beta s(y,x,z,x,z,y)
\]
for all $x,y,z \in A$.
\end{definition}

\begin{lemma}\label{lem:localglobal}
Let $\clone{C}$ be a closed oligomorphic function clone. If it has local pseudo-Siggers operations, then it has a pseudo-Siggers operation.
\end{lemma}
\begin{proof}
Let $A_0\subseteq A_1\subseteq\cdots$ be a sequence of finite subsets of $C$ whose union equals $C$, and pick for every $i\in \omega$ a $6$-ary operation $s_i\in\clone{C}$ witnessing the definition of local pseudo-Siggers operations on $A_i$, i.e., there exist unary $\alpha_i,\beta_i\in\clone{C}$ such that $\alpha_i s_i(x,y,x,z,y,z) = \beta_i s_i(y,x,z,x,z,y)$ for all $x,y,z\in A_i$. 
 Note that if  $s_i$ is such a witness for $A_i$, then so is $\gamma s_i$, for all $\gamma\in\clgr{\clone{C}}$. Hence, because $\clgr{\clone{C}}$ is oligomorphic,  we may thin out the sequence in such a way that $s_j$ agrees with $s_i$ on $A_i$, for all $j>i\geq 0$. We briefly describe this standard compactness argument for the convenience of the reader: there exists a smallest $j_0\geq 0$ such that for infinitely many $k\geq j_0$ there exists $\gamma_k\in\clgr{\clone{C}}$ such that $\gamma_k s_k$ agrees with $s_{j_0}$ on $A_0$, by oligomorphicity. Replace $s_0$ by $s_{j_0}$, all $s_k$ as above by $\gamma_k s_k$, and remove all other $s_{k'}$ where $k'\geq 0$ from the sequence. Next repeat this process picking $j_1\geq 1$ for $A_1$, and so on. This completes the argument.
 
Since the elements of the sequence $(s_i)_{i\in\omega}$ agree on every fixed $A_i$ eventually, and since $\clone{C}$ is closed, they converge to a function $s\in \clone{C}$. The function $s$, restricted to any $A_i$, witnesses local pseudo-Siggers operations on $A_i$, i.e., there exist unary $\alpha_i,\beta_i\in\clone{C}$ such that $\alpha_i s(x,y,x,z,y,z) = \beta_i s(y,x,z,x,z,y)$ for all $x,y,z\in A_i$. By a similar compactness argument as above, there exist
unary functions $\alpha,\beta\in\clone{C}$ such that $\alpha s(x,y,x,z,y,z) = \beta s(y,x,z,x,z,y)$ for all $x,y,z\in C$.
\end{proof}

We now consider gg-systems where the group $\clgr{\clone{C}}$ of a closed oligomorphic function clone $\clone{C}$ acts on finite powers of its domain.

\begin{lemma}\label{lem:1}
Let $\clone{C}$ be a closed oligomorphic function clone. Suppose that every gg-system $(\relstr{G},\group{G})$ where
\begin{itemize}
\item $\relstr{G}=(C^k;R)$ for some $k\geq 1$,
\item $\group{G}$ corresponds to the componentwise action of $\clgr{\clone{C}}$ on $C^k$,
\item $\relstr{G}$ contains $\relstr{K}_3$, and 
\item $R\subseteq C^{2k}$ is invariant under $\clone{C}$
\end{itemize}
has a pseudoloop. Then $\clone{C}$ has a pseudo-Siggers operation.
\end{lemma}
\begin{proof}
We show that $\clone{C}$ has local pseudo-Siggers operations and apply Lemma~\ref{lem:localglobal}. Let $A\subseteq C$ be finite, and pick $k\geq 1$ and $\tuple{a}^x,\tuple{a}^y,\tuple{a}^z\in A^k$ such that the rows of the $(k\times 3)$-matrix $(\tuple{a}^x,\tuple{a}^y,\tuple{a}^z)$ form an enumeration of $A^3$. Let $R$ be the binary relation on $C^k$ where tuples $\tuple{b},\tuple{c}\in C^k$ are related via $R$ if there exists a $6$-ary $s\in \clone{C}$ such that $\tuple{b}=s(\tuple{a}^x,\tuple{a}^y,\tuple{a}^x,\tuple{a}^z,\tuple{a}^y,\tuple{a}^z)$ and $\tuple{c}=s(\tuple{a}^y,\tuple{a}^x,\tuple{a}^z,\tuple{a}^x,\tuple{a}^z,\tuple{a}^y)$. 
In other words, it is the $\clone{C}$-invariant  subset of $(2k)$-tuples generated by the six vectors obtained by concatenating $\tuple{a}^u$ and $\tuple{a}^v$, where $u,v\in\{x,y,z\}$ are distinct. The latter description reveals that $R$ is a symmetric relation on $C^k$ invariant under $\clone{C}$ and containing $\relstr{K}_3$, therefore the gg-system $(R,\group{G})$, where $\group{G}$ is the componentwise action of $\clgr{\clone{C}}$ on $C^k$,  has a pseudoloop $(\tuple{b},\tuple{c})$. That means that there exists a $6$-ary $s\in \clone{C}$ and $\alpha \in\clgr{\clone{C}}$ such that $s(\tuple{a}^x,\tuple{a}^y,\tuple{a}^x,\tuple{a}^z,\tuple{a}^y,\tuple{a}^z)=\alpha s(\tuple{a}^y,\tuple{a}^x,\tuple{a}^z,\tuple{a}^x,\tuple{a}^z,\tuple{a}^y)$, proving the claim.
\end{proof}

\begin{corollary}\label{cor:K3-Siggers}
Let $\relstr{A}$ be an $\omega$-categorical core. Then either it pp-interprets $\relstr{K}_3$ with parameters, or $\Pol(\relstr{A})$ has a pseudo-Siggers operation.
\end{corollary}
\begin{proof}
We apply Lemma~\ref{lem:1} to the clone $\clone{C}:=\Pol(\relstr{A})$; then $\clgr{\clone{C}}$ consists precisely of the automorphisms of $\relstr{A}$. 
If the assumptions of this lemma are satisfied, then $\clone{C}$ has a pseudo-Siggers operation.
Otherwise, there exists a pseudoloop-free gg-system $((C^k;R),\group{G})$ satisfying the four conditions. By Lemma~\ref{lem:loop}, this gg-system pp-interprets $\relstr{K}_3$ with parameters.  Since $R$ is invariant under $\clone{C}$, it is pp-definable from $\relstr{A}$ by~\cite{BodirskyNesetrilJLC}. Moreover, since $\relstr{A}$ is a core, the orbits of $\group{G}$ are pp-definable from $\relstr{A}$ as well by~\cite{Bodirsky-HDR}. It follows that $\relstr{A}$ pp-interprets $\relstr{K}_3$ with parameters, as required.
\end{proof}

We are now ready to prove Theorem~\ref{thm:main}.

\begin{proof}[Proof of Theorem~\ref{thm:main}]
We first prove that (iii) implies (ii).
Take $\alpha,\beta,s \in \Pol(\relstr{A})$ satisfying the pseudo-Siggers identity. We claim that every stabilizer $\Pol(\relstr{A},c_1,\dots, c_n)$ has a pseudo-Siggers operation. To see that, consider the endomorphisms $\gamma,\delta$ of $\relstr{A}$ defined by $\gamma(x)=s(x,\dots,x)$, $\delta(x) = \alpha\gamma(x)$ ($=\beta\gamma(x)$ by the pseudo-Siggers identity). Because $\relstr{A}$ is a core, its automorphisms are dense in endomorphisms, thus there exist automorphisms $\epsilon, \theta$ of $\relstr{A}$ such that $\epsilon(c_i) = \gamma(c_i)$ and $\theta(c_i) = \delta(c_i)$ for every $i$. But then $\theta^{-1}\alpha\epsilon$, $\theta^{-1}\beta\epsilon$ and $\epsilon^{-1}s$ are contained in $\Pol(\relstr{A},c_1, \dots, c_n)$ and satisfy 
\[
(\theta^{-1}\alpha\epsilon)(\epsilon^{-1}s)(x,y,x,z,y,z) \approx (\theta^{-1}\beta\epsilon)(\epsilon^{-1}s)(y,x,z,x,z,y). 
\]


The implication from (ii) to (i) is trivial. 

Finally, assume that no stabilizer of $\Pol(\relstr{A})$ has a continuous homomorphism to $\trivclone$. Then no such stabilizer has a continuous clone homomorphism to $\Pol(\relstr{K}_3)$ since it is well-known that the latter clone has a continuous homomorphism to $\trivclone$.
By Theorem~\ref{thm:int}, $\relstr{A}$ does not pp-interpret $\relstr{K}_3$ with parameters. Corollary~\ref{cor:K3-Siggers} then tells us that $\Pol(\relstr{A})$ has a pseudo-Siggers operation.
\end{proof}





\section{Discussion}\label{sect:disc}

Our main theorem can be used as a tool for proving hardness: If an $\omega$-categorical core structure $\relstr{A}$ does not have a pseudo-Siggers polymorphism, then $\relstr{A}$ interprets all finite structures with parameters by the combination of Theorem~\ref{thm:main} and Theorem~\ref{thm:int}, and therefore $\CSP(\relstr{A})$ is $\NP$-hard. The pseudoloop lemma itself can be regarded as a hardness result as well:

\begin{corollary}
Let $\relstr{A}$ be an $\omega$-categorical core structure.
If $\relstr{A}$ pp-defines with parameters a pseudoloop-free graph containing a $\relstr{K}_3$, then $\CSP(\relstr{A})$ is $\NP$-hard.
\end{corollary}

Recall that a digraph is \emph{smooth} if each vertex has an incoming and an outgoing edge, and a digraph has \emph{algebraic length 1} if it contains a closed walk with one more forward edges than backward edges. 
The finite loop lemma for graphs~\cite{HellNesetril,B05} has been generalized to finite smooth digraphs with algebraic length $1$ in~\cite{BartoKozikNiven,Cyclic}. We conjecture that the pseudoloop lemma can be generalized to such digraphs as well.

\begin{conjecture}~\label{conj:loop}
Let $\group{G}$ be an oligomorphic permutation group on $G$ and let $\relstr{G}$ be a countable smooth digraph of algebraic length $1$ on $G$ which is  invariant under $\group{G}$. Then either $\relstr{G}$ contains a pseudoloop, or $\relstr{G}$ together with the orbits of $\group{G}$ on finite tuples pp-interprets $\relstr{K}_3$ with parameters.
\end{conjecture}

As intermediate steps we suggest the following stronger assumptions on $\relstr{G}$.
\begin{itemize}
\item[(a)] $\relstr{A}$ is a non-bipartite graph;
\item[(b)] $\relstr{A}$ is a digraph containing 
\[
(\{a,b,c\}; \{(a,b),(b,c),(c,a),(b,a)\})
\] 
as a subgraph (not necessarily induced);
\item[(c)] $\relstr{A}$ is a strongly connected digraph of algebraic length $1$ (equivalently, the GCD of the length of cycles is $1$).
\end{itemize}
A positive answer to Conjecture~\ref{conj:loop} under the assumption (b) or (c) would allow a strengthening of item (iii) of Theorem~\ref{thm:main} to a $4$-variable pseudo-Siggers operation $\alpha s(r,a,r,e) \approx \beta s(a,r,e,a)$ (see~\cite{Sig10,KMM14}).
Another open problem is whether it is possible to replace item (iii) of Theorem~\ref{thm:main} by pseudo-weak-near-unanimity operations (see~\cite{MM08}).
On the negative side, it has been observed that the CSP classification for the  reducts of $(\relstr{Q};<)$ shows that
the syntactically strongest characterization of (iii) in Theorem~\ref{thm:main} in the finite case by means of cyclic operations~\cite{Cyclic} cannot be lifted to the infinite, at least not in the straightforward way of adding unary functions.

Our final remarks concern a new dichotomy conjecture that has been formulated for CSPs of reducts of finitely bounded homogeneous structures~\cite{wonderland}. An \emph{h1 clone homomorphism} is a mapping from one function clone to another which preserves arities and composition with projections (equivalently, preserves identities of height $1$). When the polymorphism clone of an $\omega$-categorical structure has a uniformly continuous h1 clone homomorphism to $\trivclone$, then its CSP is $\NP$-hard. The conjecture states that the reducts of finitely bounded homogeneous structures whose polymorphism clone does not have such a mapping have  polynomial-time solvable CSP. 
The following corollary summarizes all we know so far about various clone homomorphisms to the clone of projections. In particular, it shows that the new conjecture is ``better'' in that it is implied by Conjecture~\ref{conj:dicho} (but not necessarily vice-versa).

\begin{corollary}
Consider the following statements for an $\omega$-categorical core $\relstr{A}$.
\begin{enumerate}
\item[(1)] $\Pol(\relstr{A})$ has a uniformly continuous clone homomorphism to $\trivclone$.
\item[(1')] $\Pol(\relstr{A})$ has a continuous clone homomorphism to $\trivclone$.
\item[(2)] $\Pol(\relstr{A})$ has a clone homomorphism to $\trivclone$.
\item[(3)] Some $\Pol(\relstr{A},c_1,\ldots,c_n)$ has a clone homomorphism to $\trivclone$.
\item[(3')] Some $\Pol(\relstr{A},c_1,\ldots,c_n)$ has a continuous clone homomorphism to $\trivclone$.
\item[(3'')] Some $\Pol(\relstr{A},c_1,\ldots,c_n)$ has a uniformly continuous clone homomorphism to $\trivclone$.
\item[(4)] Some $\Pol(\relstr{A},c_1,\ldots,c_n)$ has a uniformly continuous  h1 clone homomorphism to $\trivclone$. 
\item[(4')] $\Pol(\relstr{A})$ has a uniformly continuous  h1 clone homomorphism to $\trivclone$. 
\item[(5)] $\Pol(\relstr{A})$ has a continuous  h1 clone homomorphism to $\trivclone$. 
\item[(5')] Some $\Pol(\relstr{A},c_1,\ldots,c_n)$ has a continuous  h1 clone homomorphism to $\trivclone$. 
\item[(6)] $\Pol(\relstr{A})$ has an  h1 clone homomorphism to $\trivclone$. 
\item[(6')] Some $\Pol(\relstr{A},c_1,\ldots,c_n)$ has an  h1 clone homomorphism to $\trivclone$. 
\end{enumerate}
Then all statements with equal number are equivalent, and (i) implies (j) for all $1\leq i\leq j\leq 6$.
\end{corollary}
\begin{proof}
It is clear that the strength of the statements is decreasing; the only non-trivial part are the equivalences.
(1) and (1') are equivalent by the proof in~\cite{Topo-Birk} (cf.~\cite{GPin15} for an explicit proof thereof). (3) and (3') are equivalent by Theorem~\ref{thm:main}, and (3') and (3'') again by the proof in~\cite{Topo-Birk}. (4) and (4'), (5) and (5'), as well as (6) and (6') are equivalent by~\cite{wonderland}.
\end{proof}

We remark that (4) is the weakest condition known to imply NP-hardness of the CSP. 

An example communicated to us by Ross Willard shows that the implication from (2) to (3) cannot be reversed, even for finite $\relstr{A}$; another example is known among the reducts of $(\mathbb Q;<)$. For all remaining implications, no counterexamples are known. The implication from (2) to (1) is conjectured in~\cite{BPP-projective-homomorphisms}, the implication from (4) to (3) would, if true, imply that the two conjectures are equivalent. 
The most optimistic speculation would be that (6) implies (3). A positive answer would show that both conjectures can be true and yield a particularly esthetically pleasing --- core-free, topology-free, and without identities of height greater than $1$ --- formulation of the tractability conjecture.

\bibliographystyle{alpha}

\bibliography{pseudo}





\end{document}